\documentclass[10pt]{article}
\usepackage[margin=1in]{geometry}
\usepackage{amsmath,amssymb,amsthm}
\usepackage{booktabs}
\usepackage{graphicx}
\usepackage{pgfplots}
\pgfplotsset{compat=1.17}
\usepackage[numbers,sort&compress]{natbib}
\usepackage[hidelinks]{hyperref}
\usepackage{microtype}

\newtheorem{theorem}{Theorem}

\newcommand{\bpt}{\textsc{bpt}}

\title{Budgeting Bytes: A Windowed Storage Roofline and Dual-Budget\\Architecture Ablations for Storage-Bound LLM Decoding}
\author{Hanhaodi Zhang\\\small \texttt{a313295747@gmail.com} --- draft for workshop submission}
\date{}

\begin{document}
\maketitle

\begin{abstract}
Autoregressive decoding on cheap hardware is bound not by FLOPs but by the bytes each generated token must move across the slowest populated tier of a memory hierarchy. We treat \textbf{bytes-per-token (\bpt)} as a first-class design axis, organized by an \emph{address-determinism taxonomy}: parameters are classified by when their fetch address becomes known during a token's forward pass (A0: at token sampling; A1: before attention; A2: layerwise data-dependent; A3: always read). The classification reduces prefetch scheduling to single-machine feasibility with release times, yielding a closed-form \emph{windowed roofline}: a full-hiding criterion and exposed-latency bound via EDF optimality. Address-determinism is orthogonal to value-determinism: it classifies which \emph{bytes to fetch}, so Gemma~3n's per-layer embeddings and MoLE's lookup experts share one A0 access pattern (differing only in a resident, context-dependent gate); pre-attention routing is an A1 widening of A2 MoE. We validate the model against a published microcontroller deployment ($<3\%$ calibration error) and run dual-budget (\bpt{}$\times$storage) ablations across three sub-100M scales (23 training runs): (1) pre-attention routing costs no measurable loss (four of four seed-paired comparisons within noise); (2) a four-point dense-fit inversion prices routed capacity at $\gamma_1\!\ll\!1$ of resident capacity at this scale; (3) the extreme-\bpt{} A0 corner is Pareto-dominated here by a smaller A1 model; (4)~the quality--\bpt{} frontier is three-regime; (5) under a \bpt{} constraint, over-training a small conditional-capacity model is competitive with scaling up. We then take the framework to real large-MoE deployment and report a substantial \emph{negative} result that its own roofline predicts: on an 8\,GB RK3588 running Qwen3-30B-A3B (18\,GB, 4-bit), the model overflows RAM and decode is pinned at the eMMC bandwidth ceiling ($0.19$--$0.51$\,tok/s); \emph{predictive expert prefetch does not help}---not temporal-locality prefetch (net-negative, $0.19\!\to\!0.12$), not even a trace-driven \emph{oracle} (perfect prediction, $0.12\!\to\!0.13$), because the binding constraint is byte volume over a saturated bus, which prefetch cannot reduce. The lever that works is reducing \bpt{} so the model fits the fast tier: quantized to fit a 16\,GB unified-memory device, the same model runs GPU-resident at 11.5\,tok/s ($22\times$). We reconcile this with GPU-serving predictors (ProMoE et al.): a frozen-model probe predicts Qwen3-30B routing from the pre-attention state at $91.2\%$ (a scale-invariant \emph{predictability} property), but this converts to throughput only where the fast tier caches most of the model and transfer$\,\approx\,$compute---measured to hold on an A100 PCIe-offload path and to fail on bandwidth-walled edge storage. Predictability is not speedup; we chart exactly where the gap closes.
\end{abstract}

\section{Introduction}
The dominant design tradition for language models assumes weights reside in high-bandwidth memory. The most numerous computing devices---toys, appliances, wearables, white-label single-board computers---invert the economics: flash is nearly free while RAM is the scarce line item. A recent hobbyist deployment makes the gap concrete: a 28.9M-parameter model decodes at 9.5\,tok/s on an \$8 ESP32-S3 \emph{only because} 25M of its parameters sit in flash and are touched ${\sim}450$ bytes per token \citep{slvdev2026}; read naively, the same weights would yield 0.4\,tok/s.

Single-stream decoding performs roughly two operations per weight byte; on virtually every platform it is bandwidth-bound. Throughput is therefore governed not by parameter count but by \textbf{bytes-per-token}---and, when bytes traverse a storage hierarchy, by \emph{which bytes can be prefetched early enough to hide}. The regime is not a curiosity: the MiniMind framework for training 25M--200M models from scratch has drawn a 50k-star practitioner community and ships an MoE variant whose 198M-total/64M-active design is a storage-vs-\bpt{} trade made without a theory to guide it. Neither FLOP-indexed scaling laws \citep{hoffmann2022} nor MoE scaling studies conducted at $\geq$100M in HBM \citep{clark2022,krajewski2024} answer the questions this regime poses.

\paragraph{Contributions.}
(1)~An \emph{address-determinism taxonomy} and a \emph{windowed storage roofline} (\S\ref{sec:theory}): prefetch scheduling reduces to preemptive single-machine feasibility with release times, giving a closed-form hiding criterion (Theorem~\ref{thm:hiding}); the taxonomy unifies per-layer embeddings \citep{gemma3n}, lookup experts \citep{mole2025}, and pre-attention routing \citep{smallthinker}.
(2)~\emph{Dual-budget ablations at sub-100M scale} (\S\ref{sec:results}): 23 runs across three scales under joint (\bpt{}, storage) budgets, seed-paired where deltas are small, in a regime where prior conclusions are unverified extrapolations.
(3)~\emph{A validated open toolchain}: a simulator calibrated to a real MCU-class deployment ($<3\%$ error), microbenchmarks showing ternary quantization's compute dividend is only $1.2\times$ on multiplier-equipped cores (its value is byte halving), and overlap measurements confirming I/O--compute parallelism ($<6\%$ interference).

\section{Related Work}
\textbf{Capacity via cheap-address lookups.} Gemma~3n parks per-layer, per-token-id vectors in flash \citep{gemma3n}; MoLE re-parameterizes embedding-fed experts into token-id-indexed LUTs \citep{mole2025}; MoLKV adds context-aware key--value experts \citep{molkv2025}. None model bandwidth, prefetching, or multi-tier storage; MoLE's evaluation assumes 16\,GB/s PCIe. We show PLE and MoLE-LUTs share one A0 \emph{access pattern}---both fetch a token-id-addressed table row, knowable at sampling---while differing in \emph{value} determinism: MoLE's output is a fixed token-id map, whereas PLE applies a resident, context-dependent gate to the fetched row. Address-determinism (what to prefetch) is thus orthogonal to value-determinism (what the row computes); the taxonomy classifies the former. We also quantify the storage inflation precomputation implies at small scale.
\textbf{MoE under constraints.} SmallThinker trains a pre-attention router so expert fetch overlaps attention \citep{smallthinker}; EdgeMoE manages expert swap-in \citep{edgemoe}; an empirical study finds MoE \emph{slower} than dense on Jetson-class hardware due to footprint and dispatch \citep{edgemoestudy}. Theorem~\ref{thm:hiding} explains these as window-structure facts.

\textbf{Production static offload.} llama.cpp's \texttt{-{}-n-cpu-moe}/\texttt{-{}-override-tensor} places expert FFN tensors in slow memory while keeping attention in fast memory --- the A3-fast/A1-slow split of our taxonomy, arrived at as an engineering heuristic ("attention is small and hot; experts are big and cold"). It runs Qwen3-30B-A3B at ${\sim}30$\,tok/s on 6\,GB VRAM + 32\,GB RAM. Our contribution relative to this widely-used baseline is threefold: (i) the theory that says \emph{why} this placement is near-optimal and \emph{where} it breaks (the prefill wall, when every token touches disjoint experts); (ii) an A0 tier (PLE) and address-tier-dependent quantization it has no notion of; and (iii) prediction --- \texttt{-ot} places experts statically and reads them on demand, exactly the ``two-tier expert cache'' that is an open llama.cpp feature request (\#20757).

\textbf{Expert-prefetch prediction.} A parallel line predicts expert activation to hide offload latency: Pre-gated MoE co-designs a next-layer gate \citep{pregatedmoe}, SiDA hashes from embeddings \citep{sida}, ProMoE and cross-layer gating learn predictors reaching $80$--$85\%$ accuracy \citep{promoe,crosslayergate}. All target GPU-VRAM caching at $\geq$10B scale and report end-to-end speedups (${\sim}1.3$--$2\times$ over LRU); none is in the CPU/edge streaming path where \texttt{-ot} operates. We supply the missing measurement (\S\ref{sec:hint}): a same-model/same-data \emph{paired} A0-vs-A1 probe isolating predictability as a function of filtration depth, at sub-100M scale, folded into the hiding theorem --- directly the predictor the \#20757 cache would need.
\textbf{Scaling laws.} Effective parameter counts for routed models \citep{clark2022} and granularity-aware MoE laws \citep{krajewski2024} operate FLOP-side at $\geq$100M. We price the EPC \emph{per byte class}, two orders of magnitude smaller.
\textbf{Quantization.} GPTQ/AWQ deliver 4-bit PTQ \citep{gptq}; BitNet~b1.58 matches fp16 from 3B upward \citep{bitnet}; PT\textsuperscript{2}-LLM ternarizes post-training at 7B+ \citep{pt2llm}. We add a negative calibration for the sub-100M storage-bound regime.

\section{Taxonomy and the Windowed Roofline}\label{sec:theory}
A token's forward pass executes a fixed compute schedule; a single preemptable I/O channel of bandwidth $B$ runs concurrently (empirically $<6\%$ mutual interference, \S\ref{sec:calib}). Weight group $g$ has address-reveal time $r_g$, use time $u_g$, and transfer time $p_g=s_g/B$. The address-determinism class fixes $r_g$:
A0 (known at token sampling: embedding rows, PLE tables, LUT experts; window spans the whole prefix), A1 (known before layer-$l$ attention: pre-attention-routed experts; window is one attention block), A2 (known only after attention: standard MoE; window ${\approx}\,\emptyset$), A3 (always read: attention projections, dense FFN, LM head).

\begin{theorem}[Hiding feasibility and exposed latency]\label{thm:hiding}
With preemptive transfers, all I/O is hidden iff for every interval $I$ of the compute timeline,
$\sum_{g:\,[r_g,u_g]\subseteq I} s_g \;\le\; B\cdot C(I)$,
where $C(I)$ is compute time in $I$. Otherwise the added latency equals the maximum interval deficit, and EDF attains it.
\end{theorem}
\begin{proof}[Proof sketch]
The transfer set $\{(r_g,\,d_g{=}u_g,\,p_g)\}$ is a preemptive single-machine feasibility instance with release times and deadlines; the density condition and EDF optimality are classical \citep{horn1974}. Mapping machine time to compute time uses the measured I/O--compute parallelism.
\end{proof}

\textbf{Corollaries.} (i)~A1 windows are disjoint across layers, so the criterion decomposes into per-layer inequalities $s_l\le B\,C^{\text{attn}}_l$. (ii)~A0 windows are maximal---A0 tolerates the slowest storage, and EDF auto-prioritizes mixed A0/A1 traffic. (iii)~A2 bytes are exposed on slow tiers, consistent with measured edge MoE slowdowns \citep{edgemoestudy}. (iv)~Longer context widens A1 windows: context length is the prefetcher's ally.

\textbf{Dual budgets.} A device with bandwidth $B$ and target rate $R$ imposes $\bpt\le B/R$ (minus hidden fractions); its flash imposes a storage cap. Precomputation (MoLE-style A0) buys minimal \bpt{} at maximal storage---the budgets are not interchangeable.

\textbf{Pricing capacity.} Following \citet{clark2022}, define $N_{\text{eff}}$ of a variant as the dense parameter count achieving equal loss under a dense fit $L=E+A N^{-\alpha}$. We decompose $N_{\text{eff}}\approx N_{A3}+\gamma_1 N_{A1}+\gamma_0 N_{A0}$ as a first-order ansatz and measure the exchange rates.

\section{Experimental Setup}
TinyStories \citep{tinystories} (486M tokens, 4k BPE), identical recipe for all runs (AdamW, cosine, 20k steps $\times$ batch 64 $\times$ 512 ctx ${\approx}$ 650M tokens, bf16, one A100 per run). All variants share a trunk (RoPE, RMSNorm, tied head) and differ only in the capacity mechanism: dense ($\tfrac83\times$ SwiGLU), fine-grained MoE (8 experts/layer, top-2; pre- vs.\ post-attention router), MoLE-style LUT experts, each MoE with PLE. Scales: d256/6L (${\sim}6$M), d384/8L (16--24M), d512/12L (${\sim}40$--60M). RQ2 comparisons are seed-paired (1337/42/7). \bpt{} and storage are computed analytically at 4-bit weights. Generality: the MiniMind-3 shape (d768/8L, 6.4k vocab, coarse 4$\times$top-1 MoE) on 328M tokens of its Chinese corpus.

\section{Results}\label{sec:results}
\subsection{Premises hold on real hardware}\label{sec:calib}
Concurrent full-rate streaming (3.7\,GB/s, page cache bypassed) slows saturated compute by 5.7\%; overlap pipelines match Theorem~\ref{thm:hiding} within 7\% (serial 224 vs.\ predicted 209\,ms/token; overlapped 139.6 vs.\ 130), with hiding ${\approx}100\%$ where the per-layer inequality holds. Ternary kernels: $1.20\times$ over int4 (branchless LUT); branching kernels are $6\times$ \emph{slower}. A second negative calibration: replacing SwiGLU with ReGLU (SmallThinker's activation-sparsity lever) costs $+0.008$ loss at 24M (1.3586 vs.\ 1.3508, paired) --- its reported ${>}0.6$ sparsity at 4B does not pay off at small scale, where the SiLU$\to$ReLU expressivity loss outweighs the skippable computation; ReGLU is a RAM-resident, larger-$d_{ff}$ lever, not a streaming one.

\subsection{Pre-attention routing is free}
At 24M over three paired seeds: pre-attention $1.3489\pm0.0032$ vs.\ post-attention $1.3511\pm0.0046$; paired differences $+0.0028/+0.0004/+0.0031$, all favoring pre-attention. The 60M pair replicates ($1.2475$ vs.\ $1.2479$). Four of four paired comparisons favor pre-attention, none significantly: the A1 prefetch window costs nothing at these scales, extending SmallThinker's 4B claim ${\sim}170\times$ smaller.

\subsection{Routed capacity is steeply discounted}
Equal-parameter dense (24.1M, 1.3023) beats the 23.9M MoE (1.3453). Fitting $L=E+AN^{-\alpha}$ on four dense scales (5.7/15.7/24.1/39.5M) and inverting each MoE yields $\gamma_1\in[0.12,0.35]$ across all scales and fit variants: routed bytes are worth roughly a fifth of resident bytes here, versus near-parity at $\geq$100M \citep{clark2022}. We deliberately report a range: $E$ is weakly identified on our span (fits drift to grid boundaries; the cross-scale ordering of $\gamma_1$ is fit-sensitive), so the scale trend remains open while $\gamma_1\ll1$ is robust.

\subsection{The A0 corner is Pareto-dominated here}
LUT experts cost $+0.17$--$0.19$ loss versus same-scale dense across our span (9M: $+0.171$; 24M: $+0.191$; 60M: $+0.165$) at 6--10$\times$ storage (26/53/106\,MB). At sub-100M the corner collapses: a \emph{smaller} pre-attention MoE (9M total) beats the 24M LUT variant on loss (1.497 vs.\ 1.533), \bpt{} (1.87 vs.\ 3.16\,MB), and storage (4.3 vs.\ 52.7\,MB) simultaneously. Single-axis \bpt{} optimization is actively misleading in this regime.

\subsection{Reading the frontier}
\textbf{Starved end} ($\bpt<3$\,MB): pre-attention MoE+PLE strictly dominates dense (1.4969 @ 1.87\,MB vs.\ 1.5212 @ 2.73\,MB). \textbf{Middle} (${\approx}5$\,MB): $-35\%$ \bpt{} for $+0.005$--$0.010$ loss---${\approx}1.5\times$ decode speed where bandwidth-bound. \textbf{Affluent end} ($\bpt>11$\,MB): dense retakes quality at equal parameters and holds minimal storage. Design rule: \emph{the tighter the byte budget, the more conditional (A0/A1) capacity wins; the looser, the more resident (A3) bytes dominate.}

\begin{figure}[t]\centering
\begin{tikzpicture}
\begin{axis}[width=.72\linewidth,height=5.4cm,xlabel={bytes-per-token (MB, 4-bit)},ylabel={val.\ loss},legend pos=north east,legend style={font=\scriptsize},mark size=2pt]
\addplot[only marks,mark=o,blue] coordinates {(2.73,1.5212)(7.50,1.3416)(11.48,1.3023)(18.81,1.2397)};\addlegendentry{dense (A3)}
\addplot[only marks,mark=square*,red] coordinates {(1.87,1.4969)(4.84,1.3489)(11.90,1.2475)};\addlegendentry{pre-attn MoE+PLE (A1+A0)}
\addplot[only marks,mark=triangle*,brown] coordinates {(1.30,1.6921)(3.16,1.5330)(7.41,1.4045)};\addlegendentry{LUT experts (A0)}
\end{axis}
\end{tikzpicture}
\caption{Quality--\bpt{} frontier across three scales (marker groups left$\to$right: d256, d384, d512). Storage (not shown) is 6--10$\times$ larger for the A0 series. Conditional capacity dominates the starved end; dense the affluent end.}
\end{figure}
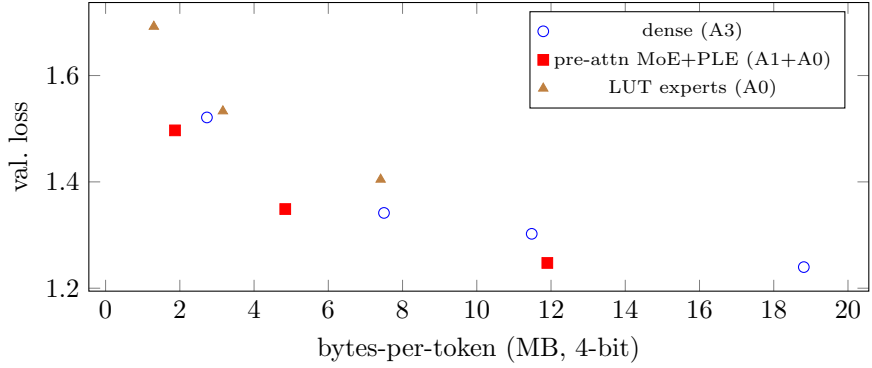

\subsection{Generality: MiniMind shape, Chinese corpus}
Replicating the MiniMind-3 shape in our trainer (single seed per cell): dense 1.9447 (\bpt{} 32.7\,MB); as-shipped A2 MoE 1.9820 (\bpt{} 32.7\,MB, storage 96.9\,MB); our pre-attention+PLE modification 1.9474 (\bpt{} 33.3\,MB, storage 102.1\,MB). The taxonomy \emph{reads} the popular config: top-1 over full-size experts makes MoE \bpt{} \emph{identical} to dense---the byte axis gains nothing, while A2 routing exposes transfers on slow storage; $3\times$ storage buys a worse loss at this budget. Our modification recovers the deficit to within noise ($\Delta$ 0.0027) while making expert bytes prefetchable. Disentangling cells separate the two factors: on this shape, moving the router pre-attention alone improves loss by $0.022$ (1.9601 vs.\ 1.9820, both without PLE) and PLE adds a further $0.013$; on the fine-grained primary setup the split is PLE $-0.012$, router ${\approx}$neutral ($+0.002$). That pre-attention routing \emph{helps} outright on the coarse top-1 shape---rather than merely costing nothing---is a suggestive single-seed observation we flag for replication.

\subsection{Real-tier deployment: an end-to-end accounting check}
We implement a C streaming engine realizing the taxonomy directly: the A3 core (attention projections, norms, router) resides in RAM; A1 experts, A0 rows, and the LM head stream per token from storage with \texttt{F\_NOCACHE} (page cache bypassed); logits match an int8-quantization-aware PyTorch reference bit-exactly. Deployed on a USB drive whose true sequential bandwidth we measured at 38\,MB/s (24\,GB file, exceeding RAM)---coincidentally the same bandwidth class as MCU quad-SPI flash---the engine streams 4.92\,MB/token and decodes at 5.20--5.23\,tok/s, both with the correctness-test weights and with the fully trained model of \S\ref{sec:overtrain}, which generates coherent TinyStories-style narratives at this rate (roughly human read-aloud speed). The naive roofline (sequential bandwidth as denominator) predicts 7.7\,tok/s, a $1.5\times$ overestimate; substituting the \emph{measured effective bandwidth} of the actual access pattern (26\,MB/s; ${\sim}40$ scattered requests per token yield a granularity efficiency $\eta\approx0.68$ on this device) predicts 5.28\,tok/s---within $1.5\%$ of measurement. Two engineering lessons doubled as theory checks: caching ${\sim}150$\,KB of per-row quantization scales eliminated seventeen tiny reads per token ($+18\%$ throughput), and merging all streamed reads into a single I/O thread restored the single-channel premise of Theorem~\ref{thm:hiding} (two competing readers had caused seek thrashing). On this platform compute (${\sim}9$\,ms) is dwarfed by I/O (${\sim}190$\,ms), so overlap gains are marginal---as the theorem's window analysis predicts for this parameter regime; the ESP32-class regime inverts the ratio.

\textbf{Second hardware point (RK3588 ARM SBC).} We cross-compiled the same engine to aarch64 and ran the over-trained model on an 8\,GB RK3588 (Cortex-A76/A55, Debian, eMMC storage) --- a mainstream edge board. It decodes at \textbf{94.5\,tok/s} generating the same coherent stories. This is the \emph{opposite} roofline corner from the USB case: the 23\,MB model fits in the 7.5\,GB RAM, so the run is compute-bound (per-token compute $\gg$ per-token I/O), and overlap yields only $+12\%$ over serial (88.5 vs.\ 79.0\,tok/s) --- exactly the small gain Theorem~\ref{thm:hiding} predicts when the hiding window already dwarfs the transfer. Measured eMMC bandwidth spans $277$\,MB/s (sequential) $\to$ $230$\,MB/s (256\,KB expert-grain, $\eta\approx0.83$) $\to$ $42$\,MB/s (4\,KB worst case, $\eta\approx0.15$), quantifying the granularity-efficiency term for this device. Together the two boards bracket the roofline: the USB tier is I/O-bound (overlap essential), the RK3588 RAM tier is compute-bound (overlap marginal) --- the model's prediction of \emph{which} regime a device sits in is borne out at both ends. The regime where a large model overflows the RK3588's RAM---the Qwen3-30B-class case---is treated as a measured study in \S\ref{sec:largemoe}, and it turns out to be the roofline's sharpest and most counter-intuitive prediction.

\subsection{Router predictability has a scale-invariant information ladder}\label{sec:hint}
The A2$\to$A1 result (\S5.2) raises a system question: can an \emph{existing} A2 checkpoint be retrofitted for prefetch without touching its weights? We freeze the coarse (4-expert, top-1) MiniMind-shape A2 model and train tiny per-layer probes to predict its top-1 routing decision from earlier signals, purely to drive prefetch --- a wrong prediction costs latency, never correctness (the original router still computes). Two probes span the filtration: an A0 probe (from the token embedding, known at sampling) and an A1 probe (from the pre-attention hidden state). Top-1 hit rate: A0 $=0.585$, A1 $=0.833$ (random $0.25$); the A1 probe climbs from $0.77$ at layer~0 to $0.91$--$0.92$ in deep layers. Via Theorem~\ref{thm:hiding}, expected exposed expert I/O falls to $(1-p)s/B$: the A1 probe hides $83\%$ of expert wait on slow storage, the A0 probe $59\%$ (over the widest window). Strikingly, these numbers match a concurrent GPU-serving predictor (ProMoE: $58.3\%$ token-based, $84.7\%$ learned) despite a $600\times$ smaller model, a different language, and $15\times$ fewer experts --- evidence that routing predictability is governed by \emph{how much of the filtration the predictor sees}, not by scale. This retrofits any A2 MoE (including the overflow-mode deployment of trillion-parameter models) into a prefetchable one at the cost of a few-KB-per-layer predictor.

\textbf{Information, not time.} The obvious zero-training alternative --- prefetch the experts the \emph{previous} token used at each layer (temporal locality, as a naive readahead would) --- fails at this scale: it hits only $18.8\%$ on our fine-grained model (8 experts, top-2) and $28.7\%$ on the coarse one (4, top-1), \emph{at or below the random baseline}. Adjacent tokens do not reuse experts; the $83\%$ predictability lives entirely in the current token's pre-attention hidden state, not in temporal inertia. This is why a learned probe --- not last-token readahead --- is the right driver for the prefetch cache, and it sharpens the information-ladder claim: predictability is a function of \emph{filtration depth}, and the depth-zero (previous-token) signal carries almost none of it here. (Large MoEs recover some temporal locality --- ProMoE's $58\%$ token-based rate --- as more experts admit semantic clustering; our small models make the information-vs-time distinction starkest.)

\textbf{Validation on a production 30B model.} We repeat the probe experiment on the real Qwen3-30B-A3B (128 experts, top-8, 48 layers), freezing the bf16 model and training a per-layer 2-layer MLP probe on continuous multilingual corpus to predict its actual expert selection from the pre-attention hidden state. Trained on 3.6k documents of real corpus, average top-8 hit rate reaches \textbf{91.2\%} ($14.6\times$ the $6.2\%$ random baseline; per-layer $0.84$--$0.97$, deepest-16 mean $0.93$). Reporting the same set-recall metric caveat as \citet{promoe}, this places Qwen3's routing among the most predictable measured: it is almost entirely determined by the pre-attention hidden state. We stress what this does and does not establish. It is a \emph{predictability} result---an information property of the routing function, scale-invariant from our 24M toy ($83\%$) to a 30B production model---and it is the prerequisite for prefetch. It is \emph{not}, by itself, a speedup: whether hiding $91\%$ of expert I/O reduces wall-clock depends entirely on whether that I/O is on the critical path, which is a property of the \emph{device}, not the predictor. \S\ref{sec:largemoe} measures both sides of that question.

\textbf{Artifact.} We ship two system pieces. (1) \textsc{Strata}, a deployment planner that takes a MoE spec and a device (RAM, slow-storage bandwidth, measured granularity efficiency $\eta$) and reports predicted tok/s for all-in-RAM / static-offload (llama.cpp \texttt{-ot}) / predict-prefetch via the windowed roofline; e.g.\ Qwen3-30B-A3B on an 8\,GB RK3588 (eMMC) is projected $0.2\!\to\!2.7$\,tok/s from static offload to prefetch. (2) A \texttt{llama\_expert\_prefetch} component compiled into llama.cpp's common library: a decode-loop \texttt{madvise(WILLNEED)} advisor over the mmap'd expert byte-ranges, driven by a learned per-layer predictor, addressing the reactive-paging gap of Issue \#20757 without touching the ggml graph (a wrong prediction only wastes readahead; the real \texttt{mul\_mat\_id} still gathers true experts).

\subsection{Over-training moves the frontier}\label{sec:overtrain}
Chinchilla optimality is a \emph{training-side} notion; a \bpt{} constraint changes the calculus. Training the 24M A1+A0 configuration for $5\times$ the token budget (100k steps, ${\sim}3.3$B tokens) improves its loss from 1.3489 to \textbf{1.2524} --- within 0.013 of the 40M dense model at its Chinchilla-scale budget (1.2397) while streaming $3.9\times$ fewer bytes per token (4.84 vs.\ 18.8\,MB). Where decode speed is the binding constraint, the winning move is not a larger model but a smaller conditional-capacity model trained far past Chinchilla --- consistent with inference-aware scaling analyses, here given a storage-bound instantiation with an end-to-end deployed artifact.

\subsection{Large-MoE deployment: a measured negative result and the fit-first lever}\label{sec:largemoe}
We deploy Qwen3-30B-A3B (48 layers, 128 experts, top-8; 4-bit, 18\,GB on disk) on an 8\,GB RK3588 (eMMC) and an Apple M4 (16\,GB unified memory), and wire a real predictive-expert-prefetch advisor into llama.cpp's decode loop (a decode-time \texttt{madvise(WILLNEED)} over the mmap'd expert byte-ranges, driven by the captured per-layer routing). Every headline number below is measured end-to-end; the arc is a chain of \emph{honest negatives} that the roofline predicts, ending in one clean positive.

\textbf{Static baseline and the context--RAM tradeoff.} The 18\,GB model overflows 8\,GB, so experts stream from eMMC per token: $0.19$\,tok/s at the model's native $40960$ context. The KV cache is a hidden \bpt{} tax here---at $40960$ context it reserves ${\sim}3.75$\,GB, starving the expert page-cache; shrinking context to $2048$ frees that RAM and lifts decode to $0.51$\,tok/s ($2.7\times$, measured). \emph{On RAM-starved overflow deployment, shorter context is faster}---the opposite of Corollary~(iv)'s prefill-time reading, because here the binding resource is cache RAM, not the attention window.

\textbf{Prefetch does not help---even a perfect oracle.} Temporal-locality prefetch (the zero-training policy of llama.cpp issue \#20757) is \emph{net-negative}: $0.19\!\to\!0.12$\,tok/s, at a measured $46.7\%$ top-8 hit rate---the $53\%$ wrong prefetches flood a bandwidth-saturated bus and evict useful pages (the same failure mode as blind $2$\,MB readahead, $0.32\!\to\!0.14$). To remove the predictor as a confound we ran a \emph{trace-driven oracle}: a first pass records the exact experts each token uses, a second (greedy, identical tokens) prefetches them perfectly one token ahead. Perfect prediction yields \textbf{no speedup} ($0.12\!\to\!0.13$). The reason is dimensional: at ctx-2048 the device already reads at ${\sim}277$\,MB/s---the eMMC's sequential ceiling---so it is bandwidth-bound, and prefetch changes \emph{when} and \emph{how} bytes are read, never \emph{how many}. A per-token expert read is ${\sim}500$--$900$\,MB; no schedule beats $\text{bytes}/B$. A microbenchmark confirms there is no granularity headroom either: at the 2.6\,MB expert grain, scattered reads (237\,MB/s) already match sequential (247--294\,MB/s), so the \emph{sequential-expert-layout} optimization caps at $1.24\times$ and is not worth a model rewrite.

\textbf{The lever that works: fit the fast tier.} The only way to raise the ceiling is to cut \bpt{} so the model resides in the fast tier. Quantized to Q2\_K (11\,GB) it fits the M4's 16\,GB unified memory; decode then runs \emph{GPU-resident} (Metal) at \textbf{11.5\,tok/s}---a $22\times$ jump over the eMMC-overflow baseline, and coherent. The transition is clean and one-dimensional: the same model at 4-bit (18\,GB) overflows the 16\,GB device, so the GPU cannot hold it and Metal gives $1.2$\,tok/s (no better than CPU); at 2-bit it fits and flies. Every failed I/O optimization above and every succeeding one here move the same quantity---the fraction of the model resident in the fast tier.

\textbf{Reconciling with GPU-serving predictors.} Why do prefetch predictors (ProMoE et al.) report $1.3$--$2\times$ where ours yields nothing? Because their win is \emph{cache-hit reduction of transferred bytes}, not overlap, and it requires a fast tier large enough to cache most of the model plus compute comparable to transfer. We measure the governing ratio directly on an A100 (PCIe Gen4, $25.4$\,GB/s H2D measured) under CPU$\to$GPU expert offload: single-stream, transfer dominates compute and prefetch-overlap tops out at $1.34\times$; as batch grows, per-token compute rises toward the (amortized) transfer time and the overlap window opens---the ProMoE regime. The unifying variable is the fast-tier fraction $\phi$ (VRAM- or RAM-resident share of the model): prediction pays only in the band where $\phi$ is large but incomplete \emph{and} the tier boundary is on the critical path. Commodity edge ($\phi\!\approx\!0.17$, bandwidth-walled) sits below this band; a fitted model ($\phi\!=\!1$) sits above it, needing no prediction at all. A high-accuracy router predictor is therefore a \emph{conditional} systems tool, not an unconditional speedup---and \S\ref{sec:hint}'s $91.2\%$ is its enabling condition, not its guarantee.

\section{Limitations}
Sub-100M ablations use one recipe per corpus and loss as sole metric; dense seed variance ($\pm0.013$, $n{=}2$) exceeds MoE's, so cross-architecture deltas below ${\sim}0.01$ need paired designs (used for RQ2); the compute-time model is single-anchor calibrated; $\gamma$ linearity is a first-order ansatz specific to one granularity (8$\times$top-2), and its estimate is a range, not a point. The large-MoE deployment numbers (\S\ref{sec:largemoe}) are single-run per configuration; the Q2\_K fit result establishes throughput and coherence, not task quality (2-bit degradation is unmeasured), and the fitted-tier win will hold only for models a device's fast tier can accommodate. The predictability/speedup boundary is charted with one A100 offload path and one edge board; the intermediate band ($\phi$ large but incomplete, transfer${\approx}$compute) where a learned predictor pays is argued from measured endpoints rather than swept end-to-end. Finally, the roofline addresses autoregressive decoding---non-autoregressive workloads (TTS, diffusion) fall outside Theorem~\ref{thm:hiding}.

\section{Conclusion}
Storage-bound decoding acquires a design theory once bytes-per-token and address determinism are first-class: a scheduling theorem says which bytes can hide, an exchange rate says what each byte class buys, and two budget numbers place any device on a precomputed frontier. Taken to a real large-MoE deployment, the same theory delivers a sharp, useful negative: when a model overflows the fast tier onto a bandwidth-saturated bus, no prefetch schedule---not even a perfect oracle---helps, because the bottleneck is byte volume, not access pattern; the only lever is reducing bytes-per-token until the model is resident, after which ordinary compute accelerators do the rest ($22\times$, measured). A high-accuracy routing predictor is a genuine property of MoE architectures and a genuine systems tool, but only in the intermediate regime where the fast tier caches most of the model and the tier boundary is on the critical path. Knowing which regime a device occupies---and refusing to optimize the wrong one---is what the roofline buys.

\bibliographystyle{plainnat}
\bibliography{refs}

\begin{thebibliography}{19}
\providecommand{\natexlab}[1]{#1}
\providecommand{\url}[1]{\texttt{#1}}
\expandafter\ifx\csname urlstyle\endcsname\relax
  \providecommand{\doi}[1]{doi: #1}\else
  \providecommand{\doi}{doi: \begingroup \urlstyle{rm}\Url}\fi

\bibitem[bit(2025)]{bitnet}
Bitnet: 1-bit pre-training for large language models.
\newblock \emph{JMLR}, 26, 2025.

\bibitem[cro(2025)]{crosslayergate}
Accurate expert predictions in moe inference via cross-layer gate.
\newblock \emph{arXiv:2502.12224}, 2025.

\bibitem[edg(2026)]{edgemoestudy}
Does mixture-of-experts actually help inference on consumer and edge hardware?
  an empirical study.
\newblock \emph{arXiv:2606.21428}, 2026.

\bibitem[Clark et~al.(2022)Clark, de~las Casas, Guy, et~al.]{clark2022}
Aidan Clark, Diego de~las Casas, Aurelia Guy, et~al.
\newblock Unified scaling laws for routed language models.
\newblock In \emph{ICML}, 2022.
\newblock arXiv:2202.01169.

\bibitem[Du et~al.(2024)]{sida}
Zhixu Du et~al.
\newblock Sida-moe: Sparsity-inspired data-aware serving for efficient and
  scalable large mixture-of-experts models.
\newblock In \emph{MLSys}, 2024.
\newblock arXiv:2310.18859.

\bibitem[Eldan and Li(2023)]{tinystories}
Ronen Eldan and Yuanzhi Li.
\newblock Tinystories: How small can language models be and still speak
  coherent english?
\newblock \emph{arXiv:2305.07759}, 2023.

\bibitem[Frantar et~al.(2023)Frantar, Ashkboos, Hoefler, and Alistarh]{gptq}
Elias Frantar, Saleh Ashkboos, Torsten Hoefler, and Dan Alistarh.
\newblock Gptq: Accurate post-training quantization for generative pre-trained
  transformers.
\newblock In \emph{ICLR}, 2023.
\newblock arXiv:2210.17323.

\bibitem[{Google}(2025)]{gemma3n}
{Google}.
\newblock Gemma 3n model overview.
\newblock \url{https://ai.google.dev/gemma/docs/gemma-3n}, 2025.

\bibitem[Hoffmann et~al.(2022)]{hoffmann2022}
Jordan Hoffmann et~al.
\newblock Training compute-optimal large language models.
\newblock \emph{arXiv:2203.15556}, 2022.

\bibitem[Horn(1974)]{horn1974}
W.~A. Horn.
\newblock Some simple scheduling algorithms.
\newblock \emph{Naval Research Logistics Quarterly}, 21\penalty0 (1):\penalty0
  177--185, 1974.

\bibitem[Hwang et~al.(2024)Hwang, Wei, et~al.]{pregatedmoe}
Ranggi Hwang, Jianyu Wei, et~al.
\newblock Pre-gated moe: An algorithm-system co-design for fast and scalable
  mixture-of-expert inference.
\newblock In \emph{ISCA}, 2024.
\newblock arXiv:2308.12066.

\bibitem[Jie et~al.(2025)Jie, Tang, Han, et~al.]{mole2025}
Shibo Jie, Yehui Tang, Kai Han, et~al.
\newblock Mixture of lookup experts.
\newblock In \emph{ICML}, 2025.
\newblock arXiv:2503.15798.

\bibitem[Krajewski et~al.(2024)Krajewski, Ludziejewski, et~al.]{krajewski2024}
Jakub Krajewski, Jan Ludziejewski, et~al.
\newblock Scaling laws for fine-grained mixture of experts.
\newblock In \emph{ICML}, 2024.
\newblock arXiv:2402.07871.

\bibitem[{slvDev}(2026)]{slvdev2026}
{slvDev}.
\newblock esp32-ai: A 28.9m llm on an esp32-s3.
\newblock \url{https://github.com/slvDev/esp32-ai}, 2026.

\bibitem[{SmallThinker Team (SJTU IPADS and Zenergize AI)}(2025)]{smallthinker}
{SmallThinker Team (SJTU IPADS and Zenergize AI)}.
\newblock Smallthinker: A family of efficient large language models natively
  trained for local deployment.
\newblock \emph{arXiv:2507.20984}, 2025.

\bibitem[Song et~al.(2024)]{promoe}
Xiaoniu Song et~al.
\newblock Promoe: Fast moe-based llm serving using proactive caching.
\newblock \emph{arXiv:2410.22134}, 2024.

\bibitem[Wang(2025)]{molkv2025}
Zongcheng Wang.
\newblock Mixture of lookup key-value experts.
\newblock \emph{arXiv:2512.09723}, 2025.

\bibitem[Yan et~al.(2025)]{pt2llm}
Xianglong Yan et~al.
\newblock Pt$^2$-llm: Post-training ternarization for large language models.
\newblock \emph{arXiv:2510.03267}, 2025.

\bibitem[Yi et~al.(2023)]{edgemoe}
Rongjie Yi et~al.
\newblock Edgemoe: Empowering sparse large language models on mobile devices.
\newblock \emph{arXiv:2308.14352}, 2023.

\end{thebibliography}
\end{document}